\documentclass[journal]{IEEEtran}
\usepackage{booktabs}
\usepackage{threeparttable}
\usepackage{amsmath,graphicx,cite,amssymb,amsthm}
\usepackage{amsfonts,multirow,bm,array,setspace,stfloats}

\usepackage{graphicx,float,cite,amssymb}
\usepackage{graphics}
\usepackage{extarrows}
 \DeclareGraphicsExtensions{.pdf,.jpeg,.png,.jpg}   
\usepackage{multirow,bm,bbm,array,setspace,dsfont}
\usepackage{textcomp}
\usepackage{psfrag}
\usepackage{color}
\usepackage{pstricks,enumerate}
\usepackage{bbm}
\usepackage[labelformat=simple]{subcaption}

\usepackage[labelformat=simple]{subcaption}

\usepackage{algorithm,algpseudocode}
\makeatletter
\newcommand{\distas}[1]{\mathbin{\overset{#1}{\kern\z@\sim}}}%
\newsavebox{\mybox}\newsavebox{\mysim}
\newcommand{\distras}[1]{%
  \savebox{\mybox}{\hbox{\kern1pt$\scriptstyle#1$\kern1pt}}%
  \savebox{\mysim}{\hbox{$\sim$}}%
  \mathbin{\overset{#1}{\kern\z@\resizebox{\wd\mybox}{\ht\mysim}{$\sim$}}}%
}
\makeatother
\ifCLASSINFOpdf
\else
\fi

\newtheorem{lemma}{Lemma}

\allowdisplaybreaks[4]

\begin{document}
%
\title{Intelligent Surface Assisted Radar Stealth Against Unauthorized ISAC}
\author{
	\IEEEauthorblockN{
	Fan Xu, Wenhai Lai, \IEEEmembership{Graduate Student Member,~IEEE}, and Kaiming Shen, \IEEEmembership{Senior Member,~IEEE}
} 
\vspace{-1em}
\thanks{
Fan Xu is with the College of Electronic and Information Engineering, Tongji University, Shanghai,
China (email: xxiaof999@tongji.edu.cn).

Wenhai Lai and Kaiming Shen are with the School of Science and Engineering, The Chinese University of Hong Kong, Shenzhen, China (e-mails: wenhailai@link.cuhk.edu.cn; shenkaiming@cuhk.edu.cn).
}
}

%


\maketitle

\begin{abstract}
The integration of radar sensors and communication networks as envisioned for the 6G wireless networks poses significant security risks, e.g., the user position information can be released to an unauthorized dual-functional base station (DFBS). To address this issue, we propose an intelligent surface (IS)-assisted radar stealth technology that prevents adversarial sensing. Specifically, we modify the wireless channels by tuning the phase shifts of IS in order to protect the target user from unauthorized sensing without jeopardizing the wireless communication link. In principle, we wish to maximize the distortion between the estimated angle-of-arrival (AoA) by the DFBS and the ground truth given the minimum signal-to-noise-radio (SNR) constraint for communication. Toward this end, we propose characterizing the problem as a game played by the DFBS and the IS, in which the DFBS aims to maximize a particular utility while the IS aims to minimize the utility. Although the problem is nonconvex, this paper shows that it can be optimally solved in closed form from a geometric perspective. According to the simulations, the proposed closed-form algorithm outperforms the  baseline methods significantly in combating unauthorized sensing while limiting the impacts on wireless communications.
\end{abstract}

\section{Introduction}
The integration of radar sensors and communication networks is an emerging technology for the 6G wireless network. Although the integrated sensing and communications (ISAC) technology is extensively pursued in the next-generation network, not every user terminal is willing to be a part of it because of the privacy concerns. Some user terminals may only want to access the communication service without their positions being sensed by the dual-functional base station (DFBS).
The existing approaches include electromagnetic stealth \cite{Husnain2019} and spoofing \cite{Fang2024}. Electromagnetic stealth reduces the signal reflection by covering the target with electromagnetic wave-absorbing materials, while the spoofing technique generates interfering signals to mislead the radar detection. 


Although one could prevent the DFBS from sensing the target user by  sending an interfering co-spectrum signal or stealth coating, the intelligent surface (IS) is a much more economical option because it is a simple passive device without any RF chains or high-cost wave-absorbing materials \cite{Wu2021}. Meanwhile, emitting an active interfering signal can violate the radiocommunication rule or even commit a criminal offense, whereas using the IS to reflect incident signals is utterly legitimate. The notion of IS can be traced back to 2010s \cite{Wu2021}. An IS is a planar surface composed of a (large) array of cheap reflective elements, each adjusting the phase shift of its induced reflected channel. The IS was initially proposed to enhance the wireless communication only \cite{Wu2021}. Many more recent works consider using the IS to enhance radar sensing along, namely the IS-assisted ISAC \cite{Liu2023,Rihan2024,Yuan2024,Xing2023,Shao2022}. Among these works, a self-sensing IS architecture is proposed in \cite{Shao2022}, where sensors are installed to detect the target position through signals sent by the IS controller. The research on using the IS to suppress sensing is still at an early stage, with only a limited number of works in the area. For instance, \cite{Zheng2024} suggests using the IS to reduce the reflected radar signal power by the Lagrange multiplier method, and \cite{Xiong2024} considers a similar problem setup and develops a semi--closed-form solution based on the Karush-Kuhn-Tucker (KKT) condition. Differing from the above two attempts, \cite{Wang2024} uses the IS to redirect the radar echo signal and  thereby create a so-called deceptive angle-of-arrival (AoA) for the radar set. Moreover, \cite{Shao2024} considers a more complicated scenario in which the authorized and unauthorized radar sets coexist; it devises a penalty dual decomposition method for the IS phase shift optimization to hinder the radar reception at the unauthorized radar sets. While the above works all assume that the channel state information (CSI) is available more or less, another recent work \cite{Staat2022} pursues an anti-sensing IS configuration in the absence of CSI.

The present work is quite distinct from the aforementioned works in the following three respects. First, our goal is to maximize the distortion between the estimated AoA by the unauthorized DFBS and the ground truth, while previous works mostly adopt a Cram{\'e}r-Rao bound objective. Second, most existing works \cite{Zheng2024,Xiong2024,Wang2024,Shao2024} assume that the IS is deployed at the target as depicted in Fig. \ref{Fig model}(a), whereas we assume that the IS is deployed in the environment as depicted in Fig. \ref{Fig model}(b). Third, aside from mitigating the radar detection, we aim to  preserve the communication quality at the same time; this is reflected by the signal-to-noise-radio (SNR) constraint of the communication signal. The resulting optimization problem is nonconvex, but we show that it can be optimally solved in closed form after some utility-based approximation.

\section{System Model}


Consider a wireless system as shown in Fig. \ref{Fig model}(b). The DFBS is serving a user terminal in the downlink. The DFBS has one transmit antenna (TA) and $M$ radar reception antennas (RA). Assume that the DFBS attempts to detect the AoA of the user terminal without getting authorized. Thus, the DFBS plays a dual role: it is a legitimate transmitter for wireless communications and in the meanwhile is an adversary that imposes risk to the privacy of the user terminal. An IS with $N$ reflective elements is deployed somewhere in the cellular network to prevent the unauthorized sensing. Denote by $\bm{\Theta}\triangleq \textrm{diag}(\theta_1,\ldots,\theta_N)$ the phase shift array of the IS, where each $\theta_n$ is the phase-shift complex exponential with $|\theta_n|=1$, i.e., a phase shift of $\angle\theta_n$ is induced in the reflected propagation channel associated with the $n$-th reflected element. Notice that although this paper considers the passive IS without signal amplification, the proposed algorithm can be readily applied to the active IS case in which the IS is capable of amplifying the incident signals \cite{Kang2024}, so that the reflected echo signals shall be even stronger. The IS in our model is owned and controlled by the receiver, while  the DFBS is assumed to be unaware of the IS deployment. The optimization of the phase shifts at the IS is conducted at the receiver side. For the CSI acquisition, the receiver first tries out different phase shifts on the IS and then measures the corresponding received signal sent from the DFBS, thereby recovering the CSI as in \cite{Sun2023}.



\begin{figure}[!t]
\begin{minipage}[t]{0.49\linewidth}
\centering
\includegraphics[height=3.8cm]{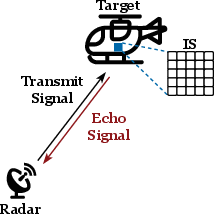}
\subcaption{ }
\end{minipage}
\begin{minipage}[t]{0.49\linewidth}
\centering
\includegraphics[height=3.8cm]{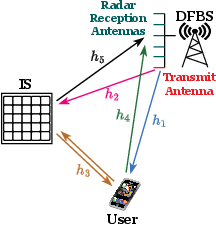}
\subcaption{ }
\end{minipage}
\caption{Two types of IS deployment.}
\label{Fig model}
\end{figure}

For ease of notation, we use the index $1$ to indicate the channel from the TA to the user terminal, the index $2$ the channel from the TA to the IS, the index $3$ the channel from the IS to the user terminal, the index $4$ the channel from the user terminal to the RA set, and the index $5$ the channel from the IS to the RA set, as shown in Fig. \ref{Fig model}(b). We assume that each transmission block consists of $L$ channel uses. For the transmit signal $\mathbf{x}\in\mathbb{C}^{L\times 1}$ from the DFBS, the received signal  $\mathbf{y}_U^\top \in\mathbb{C}^{1\times L}$ at the user terminal is given by
\begin{align}
\mathbf{y}_{U}^\top&=(h_{1}+\mathbf{h}_{3}^\top\bm{\Theta}\mathbf{h}_{2})\mathbf{x}^\top+ \mathbf{z}^\top_{U},\label{eqn model received signal}
\end{align}
where $h_{1}\in\mathbb{C}$, $\mathbf{h}_{2}\in\mathbb{C}^{N\times 1}$, $\mathbf{h}_{3}\in\mathbb{C}^{N\times 1}$, and the additive background noise $\mathbf{z}_{U}\in\mathbb{C}^{L\times 1}$. We remark that the transmit signal $\mathbf{x}\in\mathbb{C}^{L\times 1}$ carries the downlink information, and also that the DFBS uses the corresponding echo for sensing. Following the ISAC setting in \cite{Cheng2019}, we assume that $\mathbf{x}\in\mathbb{C}^{L\times 1}$ is i.i.d. with a power constraint $P$. 


Moreover, the echo heard by the DFBS can be computed as
\begin{align}
\mathbf{Y}_{R}=\zeta(\mathbf{h}_{4}+\mathbf{H}_{5}\bm{\Theta}\mathbf{h}_{3})(h_{1}+\mathbf{h}_{3}^\top\bm{\Theta}\mathbf{h}_{2})\mathbf{x}^\top+\mathbf{Z}_{R},\label{eqn model echo signal}
\end{align}
where $\mathbf{h}_{4}\in\mathbb{C}^{M\times 1}$, $\mathbf{H}_{5}\in\mathbb{C}^{M\times N}$, $\zeta$ is the reflection coefficient of the radar cross-section of the user terminal, and $\mathbf{Z}_{R}\in\mathbb{C}^{M\times L}$ is the additive background noise at the DFBS.

Further, following the previous works \cite{Xiong2024,Zheng2024,Wang2024,Shao2024,Han2022,Song2023}, we model the above channels as
\begin{align*}
&h_{1}=\alpha_{1}\xi_{1},\quad \mathbf{h}_{2}=\alpha_{2}\xi_{2}\bm{\xi}_I(\psi_{2}^{AoA} ),\quad \mathbf{h}_{3}=\alpha_{3}\xi_{3}\bm{\xi}_I(\psi_{3}^{AoD}),\notag\\
&\mathbf{h}_{4}=\alpha_{4}\xi_{4}\bm{\xi}_R(\psi_{4}^{AoA}),\quad \mathbf{H}_{5}=\alpha_{5}\xi_{5}\bm{\xi}_R(\psi_{5}^{AoA})\bm{\xi}_I(\psi_{5}^{AoD})^\top,
\end{align*}
where $\{\xi_{1},\xi_{2}$, $\xi_{3}$, $\xi_{4},\xi_{5}\in\mathbb{C}\}$ are the phase-shift complex exponentials of the different channels, 
$\{\psi_{2}^{AoA},\psi_{4}^{AoA},\psi_{5}^{AoA}\}$ are the AoA for the different channels, $\{\psi_{3}^{AoD},\psi_{5}^{AoD}\}$ are the angle-of-departure (AoD) for the different channels, $\bm{\xi}_I(\cdot)\in\mathbb{C}^{M\times 1}$ is the steering vector of the IS,  $\bm{\xi}_R(\cdot)\in\mathbb{C}^{N\times 1}$ is the steering vector of the RA set, and $\{\alpha_{1},\alpha_{2},\alpha_{3},\alpha_{4},\alpha_{5}\in\mathbb{R}\}$ are the path-loss coefficients. In particular, because the TA and the RA set are in close proximity to each other, we can assume that the associated path-loss coefficients are approximately equal, i.e., $\alpha_1\approx\alpha_4$ and $\alpha_2\approx\alpha_5$. Likewise, we can further assume that $\psi_{2}^{AoA}\approx\psi_{5}^{AoD}$ as in \cite{Liu2022}. 



\section{Problem Formulation}

Aside from the downlink transmission, the DFBS wishes to sense the AoA of the user terminal, $\psi_{4}^{AoA}$. Recall that the DFBS is unaware of the IS and its related channels, so it models the received echo signal as
\begin{align}
\tilde{\mathbf{Y}}_R = \zeta\mathbf{h}_{4}h_{1}\mathbf{x}^\top+\mathbf{Z}_{R} = \zeta\alpha_{1}^2\xi_{1}\xi_{4} \bm{\xi}_R (\psi_{4}^{AoA})\mathbf{x}^\top +\mathbf{Z}_{R},\label{eqn fake echo}
\end{align}
in contrast to the actual model in \eqref{eqn model echo signal}. Based on the echo signal $\tilde{\mathbf{Y}}_R$, the DFBS estimates $\psi_{4}^{AoA}$ jointly with $\alpha= \zeta\alpha_{1}^2\xi_{1}\xi_{4}$ by the maximum likelihood method \cite{Bekkerman2006} as
\begin{align}
(\tilde{\alpha},\tilde{\psi}_{4}^{AoA})=\arg \min_{\alpha', \psi'} \left\| \textrm{vec}(\tilde{\mathbf{Y}}_R)-  \alpha'\Big(\mathbf{I}\otimes\bm{\xi}_R(\psi')\Big)\mathbf{x}\right\|.\label{eqn ML}
\end{align}
For fixed $\psi'$, the optimal estimate of $\alpha$ in \eqref{eqn ML} is given by
\begin{align}
\tilde{\alpha}&=\frac{\mathbf{x}^H (\mathbf{I} \otimes \bm{\xi}_R({\psi}')^H )}{\|(\mathbf{I} \otimes \bm{\xi}_R({\psi}') )\mathbf{x}\|_2^2}\textrm{vec}(\tilde{\mathbf{Y}}_R).\label{eqn solution ML 1}
\end{align}
After substituting \eqref{eqn solution ML 1} into \eqref{eqn ML}, we obtain the estimate of $\psi_{4}^{AoA}$:
\begin{align}
\tilde{\psi}_{4}^{AoA}
&=\arg\max_{\psi'}|\mathbf{x}^\top\tilde{\mathbf{Y}}_R^H\bm{\xi}_R(\psi')|^2.\label{eqn solution ML 2}
\end{align}
The task of the IS is to hinder the above estimation\footnote{There are certain surveillance scenarios that enforce the target sensing, which can be addressed through legislation, e.g., by limiting the IS size.}. Thus, the IS design problem is to optimize the phase shift array $\bm\Theta$ to maximize the AoA estimation error (averaged over the random $\mathbf x$ and $\mathbf Z_R$): 
\begin{subequations}\label{eqn original problem}\begin{align}
\max_{\bm\Theta}&\quad \mathbb{E}_{\mathbf{x},\mathbf{Z}_R} \!\left[\left|\Big(\arg\max_{\psi'}|\mathbf{x}^\top {\mathbf{Y}}_R^H\bm{\xi}_R(\psi')|\Big)-\psi_{4}^{AoA}\right|\right]\label{eqn original problem 1}\\
\textrm{s.t. }&\quad\; \frac{P}{\sigma^2}|(h_{1}+\mathbf{h}_{3}^\top\bm{\Theta}\mathbf{h}_{2})|^2\ge \eta,\;|\theta_n|=1,\forall n\label{eqn original problem 2}
\end{align}\end{subequations}
where $\eta$ in \eqref{eqn original problem 2} is the minimum SNR threshold for ensuring the communication quality. 

\section{Proposed IS Optimization Method}\label{sec method}

Problem \eqref{eqn original problem} is intractable because the maximum likelihood estimator cannot be written in closed form. We first propose approximating \eqref{eqn original problem} as an inner-product utility maximizing problem. Then we show that the new problem can be solved geometrically despite its nonconvexity.


\subsection{Problem Approximation}\label{sec problem transform}

The maximum likelihood estimate of $\tilde{\psi}_{4}^{AoA}$ at the DFBS in \eqref{eqn solution ML 2} can be rewritten as
\begin{align}
&\arg\max_{\psi'}\left|\mathbf{x}^\top {\mathbf{Y}}_R^H\bm{\xi}_R(\psi')\right|\notag\\
\approx&\arg\max_{\psi'}\Big|\zeta^*\mathbb{E}\big[|[\mathbf{x}]_l|^2\big]\mathbf{g}^H\bm{\xi}_R(\psi')+\mathbb{E}\big[ [\mathbf{x}]_l[\mathbf{Z}]_l^H\bm{\xi}_R(\psi')\big]\Big|\notag\\
=&\arg\max_{\psi'}\left|(\mathbf{h}_{4}+\mathbf{H}_{5}\bm{\Theta}\mathbf{h}_{3})^H\bm{\xi}_R(\psi')\right|.\label{eqn first change 3}
\end{align}
where $\mathbf{g}= (\mathbf{h}_{4}+\mathbf{H}_{5}\bm{\Theta}\mathbf{h}_{3})(h_{1}+\mathbf{h}_{3}^\top\bm{\Theta}\mathbf{h}_{2})$. In \eqref{eqn first change 3}, the second step follows by the law of large numbers, while the last step follows by $\mathbb{E}\big[ [\mathbf{x}]_l[\mathbf{Z}]_l^H\bm{\xi}_R(\psi')\big]=0$.
Thus, the maximum likelihood estimate of $\tilde{\psi}_{4}^{AoA}$ at the DFBS is in essence to maximize the above inner product. By the above approximation, the objective function in \eqref{eqn original problem 1} becomes
\begin{align}
\left|\Big(\arg\max_{\psi'}\left|(\mathbf{h}_{4}+\mathbf{H}_{5}\bm{\Theta}\mathbf{h}_{3})^H\bm{\xi}_R(\psi')\right|\Big)-\psi_{4}^{AoA}\right|.
\label{eqn objective no E}
\end{align}
Now we assume that the DFBS learns from a ``genie'' the true $\psi_{4}^{AoA}$, and uses it to replace the estimated $\psi_{4}^{AoA}$ in \eqref{eqn objective no E}. Consequently, problem \eqref{eqn original problem} is approximated as
\begin{subequations}\label{eqn transform problem first}\begin{align}
\min_{\bm\Theta}&\quad \left|(\mathbf{h}_{4}+\mathbf{H}_{5}\bm{\Theta}\mathbf{h}_{3})^H \bm{\xi}_R(\psi_{4}^{AoA})\right|\label{eqn transform problem first 1}\\
\textrm{s.t. }&\quad\; \frac{P}{\sigma^2}|(h_{1}+\mathbf{h}_{3}^\top\bm{\Theta}\mathbf{h}_{2})|^2\ge \eta,\;|\theta_n|=1,\forall n.\label{eqn transform problem first 2}
\end{align}\end{subequations}
The above problem can be interpreted as an inner-product utility optimization, in which the new inner product $\left|(\mathbf{h}_{4}+\mathbf{H}_{5}\bm{\Theta}\mathbf{h}_{3})^H \bm{\xi}_R(\psi_{4}^{AoA})\right|$ mimics the original inner product $\left|(\mathbf{h}_{4}+\mathbf{H}_{5}\bm{\Theta}\mathbf{h}_{3})^H\bm{\xi}_R(\psi')\right|$ in \eqref{eqn first change 3} that reflects the capability of the DFBS to sense $\psi_{4}^{AoA}$. The aim of the IS is to minimize the inner-product utility under the communication constraint.

\subsection{Proposed Algorithm}

The inner-product utility problem in \eqref{eqn transform problem first} is still difficult due to its nonconvexity. Nevertheless, we show that it can be optimally solved in closed form by means of geometry.
First, introduce the variable 
\begin{equation}
\nu= \alpha_{2}\mathbf{h}_{3}^\top\bm{\Theta}\bm{\xi}_I(\psi_{2}^{AoA}),
\end{equation}
and thereby rewrite the reflective channels in \eqref{eqn model echo signal} as 
\begin{equation}
\mathbf{h}_{3}^\top\bm{\Theta}\mathbf{h}_{2}=\xi_{2}\nu\quad\text{and}\quad\mathbf{H}_{5}\bm{\Theta}\mathbf{h}_{3}=\xi_{5}\bm{\xi}_R(\psi_{5}^{AoA})\nu.
\end{equation}
Further, the objective function \eqref{eqn transform problem first 1} can be recast to
\begin{align}
&\left|(\mathbf{h}^H_{4}+\xi_{5}^*\bm{\xi}_R(\psi_{5}^{AoA})^H\nu^*)\bm{\xi}_R(\psi_{4}^{AoA})\right|.\label{eqn delete one order}
\end{align}
\begin{lemma}\label{lemma feasible nu}
The feasible region of $\nu$ in \eqref{eqn transform problem first} is given by
\begin{equation}
\mathcal{V} = \left\{\nu: |h_{1}+\xi_{2}\nu|^2\ge \frac{\eta\sigma^2}{P}, |\nu|\le N\alpha_{2}\alpha_{3}\right\}.
\end{equation}
\end{lemma}
\begin{proof}
We first show that the feasible region is a subset of $\mathcal V$. According to the definition, $\nu$ can be rewritten as
\begin{align}
\nu=\alpha_{2}\alpha_{3}\xi_{3}\sum_{n=1}^N\theta_n [\bm{\xi}_I(\psi_{3}^{AoD})]_n[\bm{\xi}_I(\psi_{2}^{AoA})]_n.\label{eqn nu circle region}
\end{align}
Since $|\xi_{3}|=|\theta_n [\bm{\xi}_I(\psi_{3}^{AoD})]_n[\bm{\xi}_I(\psi_{2}^{AoA})]_n|=1, \forall n$, $\nu$ must satisfy $|\nu|\le N\alpha_{2}\alpha_{3}$. Substituting the equality $\mathbf{h}_{3}^\top\bm{\Theta}\mathbf{h}_{2}=\xi_{2}\nu$ into \eqref{eqn original problem 2}, we  obtain the constraint $|h_{1}+\xi_{2}\nu|^2\ge \frac{\eta\sigma^2}{P}$. Therefore, the necessity is proved. 

We now further show that $\mathcal V$ is a subset of the feasible region. There always exists a set of phase shifts $\{\theta_1,\ldots,\theta_N\}$ that satisfy \eqref{eqn nu circle region} for an arbitrary $\nu\in \mathcal{V}$. When $N$ is an even number, these phase shifts can be set as $\theta_{n}=\psi_n[\bm{\xi}_I(\psi_{3}^{AoD})]^*_n[\bm{\xi}_I(\psi_{2}^{AoA})]^*_n$ for $n\in\{1,\ldots,N\}$ with 
\begin{align*}
&\psi_n =
\left\{ 
\begin{array}{ll}
e^{j(\angle \frac{\nu }{\alpha_{2}\alpha_{3}\xi_{3}}+\arccos\frac{|\nu|}{\alpha_{2}\alpha_{3}N})},& \textrm{ if $n$ is even},\\
\frac{2\nu }{\alpha_{2} \alpha_{3} \xi_{3} N}\!-\!e^{j(\angle\frac{\nu }{\alpha_{2} \alpha_{3} \xi_{3}}+\arccos\frac{|\nu |}{\alpha_{2} \alpha_{3} N}\!)},& \textrm{ if $n$ is odd}.
\end{array}\right.
\end{align*}
When $N$ is an odd number, we then let
\begin{align*}
&\psi_n =
\left\{
\begin{array}{ll}
e^{j(\angle \frac{\bar\nu }{\alpha_{2}\alpha_{3}\xi_{3}} +\arccos\frac{|\bar{\nu} |}{\alpha_{2}\alpha_{3}N})}, & \textrm{ if $n$ is even},\\
\frac{2\bar{\nu} }{\alpha_{2}\alpha_{3}\xi_{3}N}\!-\!e^{j(\angle\frac{\bar\nu }{\alpha_{2}\alpha_{3}\xi_{3}}+\arccos\frac{|\bar{\nu} |}{\alpha_{2}\alpha_{3} N})},& \textrm{ if $n$ is odd},
\end{array}\right.
\end{align*}
where $\bar{\nu}= \nu-\alpha_{2}\alpha_{3}\xi_{3}\psi_{N}$, and particularly let
\begin{align*}
& \psi_N =
\left\{
\begin{array}{ll}
e^{j(\angle\frac{\nu }{\alpha_{2}\alpha_{3}\xi_{3}} + \arccos\frac{\alpha_{2}\alpha_{3}}{2|\nu |})},  & \textrm{ if $|\nu |\le (N-1)|\alpha_{2}\alpha_{3}|$},\\
e^{j\frac{ \nu }{\alpha_{2}\alpha_{3}\xi_{3}}},  & \textrm{ if $|\nu |> (N-1)|\alpha_{2}\alpha_{3}|$}.
\end{array}\right.
\end{align*}
The proof is then completed.
\end{proof}
Moreover, we have 
\begin{multline}\label{eqn objective}
\left|(\mathbf{h}^H_{4}+\xi_{5}^*\bm{\xi}_R(\psi_{5}^{AoA})^H\nu^*)\bm{\xi}_R(\psi_{4}^{AoA})\right|^2= \\|\bm{\xi}_R(\psi_{5}^{AoA})^H\bm{\xi}_R(\psi_{4}^{AoA})|^2|\nu|^2+|\alpha_{1}|^2M^2\\
\quad +2\cdot\Re\big(\alpha_{1}^*\xi_{4}^*M\xi_{5}\bm{\xi}_R(\psi_{4}^{AoA})^H\bm{\xi}_R(\psi_{5}^{AoA})\nu\big).
\end{multline}
Define 
\begin{align}
a&=|\bm{\xi}_R(\psi_{5}^{AoA})^H\bm{\xi}_R(\psi_{4}^{AoA})|^2,\\
b&=\alpha_{1}\xi_{4}M\xi_{5}^*\bm{\xi}_R(\psi_{5}^{AoA})^H\bm{\xi}_R(\psi_{4}^{AoA}).
\end{align}
Now, in light of Lemma \ref{lemma feasible nu} and the above notation simplification, problem \eqref{eqn transform problem first} can be further rewritten as
\begin{subequations}\label{eqn transform problem second}\begin{align}
\min_\nu &\quad a|\nu|^2+2\cdot\Re(b^*\nu)\label{eqn transform problem second 1}\\
\textrm{s.t. } & \quad|h_{1}+\xi_{2}\nu|^2\ge \frac{\eta\sigma^2}{P}\label{eqn transform problem second 2}\\
&\quad |\nu|\le N\alpha_{2}\alpha_{3}.\label{eqn transform problem second 3}
\end{align}\end{subequations}
Solving the above problem can be interpreted as projecting $-\frac{b}{a}$ to the closes point inside the feasible region, as illustrated in Fig. \ref{Fig region}, where the blue and green circles denote the boundaries set by \eqref{eqn transform problem second 2} and \eqref{eqn transform problem second 3} respectively, and the feasible region is highlight in green. If the unconstrained optimal solution, $-\frac{b}{a}$, is given by the red point in Fig. \ref{Fig region},  then the optimal solution for \eqref{eqn transform problem second} will be the green point on the right, i.e., the intersection of the green circle ($|\nu|= N\alpha_{2}\alpha_{3}$) and the line segment between $-\frac{b}{a}$ and the origin point.

\begin{figure}[t]
\centering
\includegraphics[height=3.5cm]{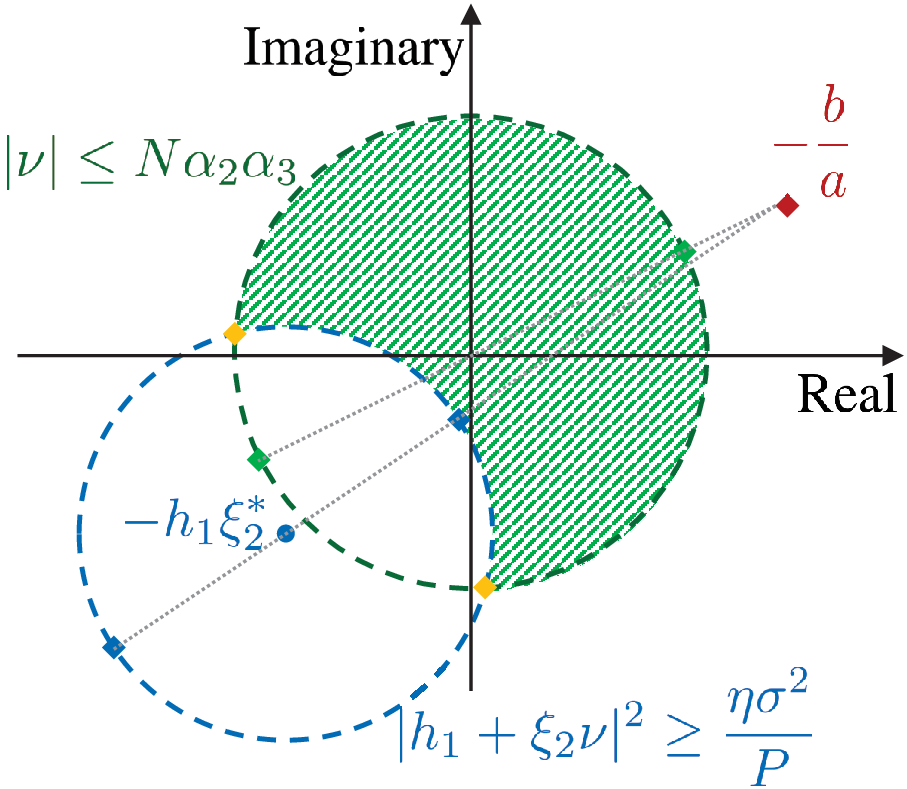}
\caption{Feasible region of $\nu$ in problem \eqref{eqn transform problem second}.}\label{Fig region}
\end{figure}

There are 6 candidates for the point closest to $-\frac{b}{a}$: $\nu_1=-\frac{b}{a}$,
$\nu_2= N \alpha_{3} \alpha_{2} e^{j(\angle  (h_{1}\xi^*_{2})+\arccos(\frac{\eta  \sigma^2/P-(N\alpha_{3}\alpha_{2})^2-|h_{1}|^2}{2N\alpha_{3}\alpha_{2} |h_{1}|}))}$, $\nu_3= N \alpha_{3} \alpha_{2} e^{j(\angle  (h_{1}\xi^*_{2})-\arccos(\frac{\eta  \sigma^2/P-(N\alpha_{3}\alpha_{2})^2-|h_{1}|^2}{2N\alpha_{3}\alpha_{2} |h_{1}|}))}$, $\nu_4=\sigma\sqrt{\frac{\eta}{P}}e^{j\angle(-b/a+h_{1}\xi^*_{2})}-h_{1}\xi^*_{2}$,
$\nu_5=-\sigma\sqrt{\frac{\eta}{P}}e^{j\angle(-b/a+h_{1}\xi^*_{2})}-h_{1}\xi^*_{2}$, $\nu_6=N\alpha_{3}\alpha_{2}e^{j\angle (-b/a)}$, $\nu_7=-N\alpha_{3}\alpha_{2}e^{j\angle (-b/a)}$.
These candidates can be geometrically explained as follows. If the unconstrained optimal solution $-\frac{b}{a}$ lies within the feasible region, then the optimal solution $\nu_{opt}$ for problem \eqref{eqn transform problem second} is $\nu_1 = -\frac{b}{a}$. If $-\frac{b}{a}$ is outside the feasible region, the optimal solution $\nu_{opt}$ for problem \eqref{eqn transform problem second} has three possible cases: (i) $\nu_{opt}$ is located at the intersection of the boundaries set by \eqref{eqn transform problem second 2} and \eqref{eqn transform problem second 3} (the yellow points in Fig. \ref{Fig region}); (ii) $\nu_{opt}$ is located at the intersection of the boundary set by \eqref{eqn transform problem second 2} and the line connecting $-\frac{b}{a}$ and $-h_{1}\xi^*_{2}$ (the blue points in Fig. \ref{Fig region}); (iii) $\nu_{opt}$ is located at the intersection of the boundary set by  \eqref{eqn transform problem second 3} and the line connecting $-\frac{b}{a}$ and the origin point (the green points in Fig. \ref{Fig region}).
The optimal solution of problem \eqref{eqn transform problem second} is the one that achieves the smallest objective value and satisfies constraints \eqref{eqn transform problem second 2} and \eqref{eqn transform problem second 3}. After obtaining the optimal solution $\nu_{opt}$, the corresponding optimal phase shifts $\bm\Theta$ can be recovered as $\theta_{n}=\psi_n[\bm{\xi}_I(\psi_{3}^{AoD})]^*_n[\bm{\xi}_I(\psi_{2}^{AoA})]^*_n$. Algorithm \ref{alg nu} summarizes the proposed method for \eqref{eqn transform problem first}.

\begin{figure}[t]
\vspace{-0.5\baselineskip}
\begin{algorithm}[H]
\caption{IS Beamforming Against Unauthorized ISAC}
\label{alg nu}
\vspace{-0.5em}
\begin{algorithmic}[1]
\State{Compute the seven candidates $\mathcal{S}=\{\nu_1,\ldots,\nu_7\}$.}
\For{$s=1,\ldots,7$}
\If{$|h_{1}+\xi_{2}\nu_s|^2< \frac{\eta\sigma^2}{P}$ or $|\nu_s|> N\alpha_{2}\alpha_{3}$}
\State{$\mathcal{S}\leftarrow\mathcal{S}\backslash \{v_s\}$}
\EndIf
\EndFor
\State{ $\nu_{opt}=\arg\min_{\nu\in\mathcal{S}}|(\mathbf{h}^H_{4}+\xi^*_5\bm{\xi}_R(\psi_{5}^{AoA})^H\nu^*)\bm{\xi}_R(\psi_{4}^{AoA})|^2$ and $\theta_{n}=\psi_n[\bm{\xi}_I(\psi_{3}^{AoD})]^*_n[\bm{\xi}(\psi_{2}^{AoA})]^*_n$.}
\end{algorithmic}
\end{algorithm}
\end{figure}

\section{Simulations}\label{sec simulation}
\begin{figure}[t]
\centering
\includegraphics[height=3.5cm]{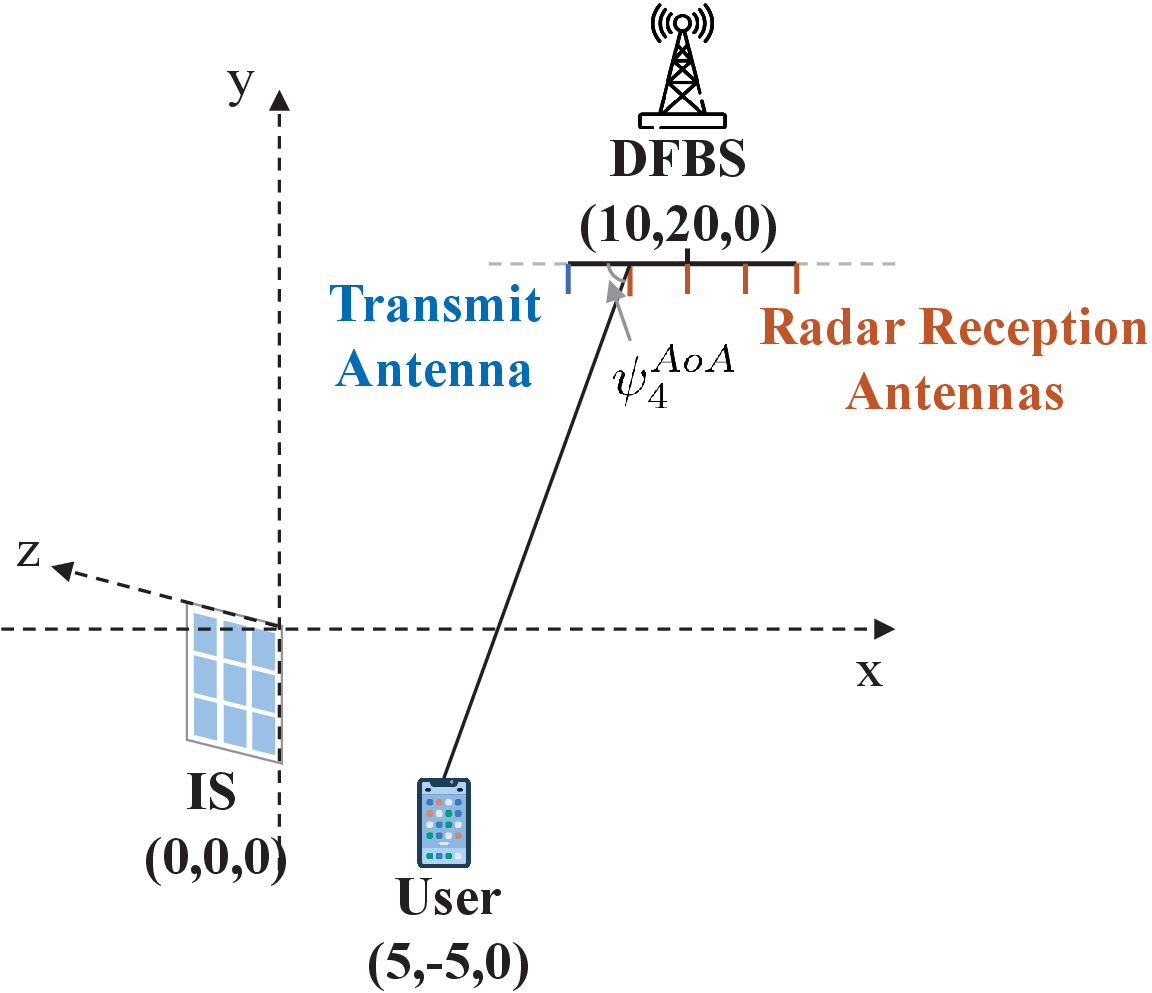}
\caption{Network topology of our simulations.}\label{Fig simulation model}
\end{figure}

This section details simulations to validate the efficacy of the proposed method. As illustrated in Fig. \ref{Fig simulation model}, the DFBS, user, and IS are located at $(10,20,0)$, $(5,-5,0)$, and $(0,0,0)$, respectively. Denote by $\{d_{1},d_{2},d_{3},d_4,d_5\}$ the distance of the different channels in Fig. \ref{Fig model}(b). Following \cite{Xu2024}, the path loss is modeled as $\alpha_i=10^{-(30+22\log_{10}d_i)/20}$ for $i\in\{1,2,3,4,5\}$. The phase-shift complex exponentials are given by 
\begin{align}
& \xi_{1}=e^{-j\frac{2\pi}{\lambda} d_{1}}e^{j\frac{2\pi}{\lambda}\epsilon_R\cos(\psi_{4}^{AoA})}, \xi_{2}\!=\!e^{-j\frac{2\pi}{\lambda}d_{2}}e^{j\frac{2\pi}{\lambda}\epsilon_R\cos(\psi_{5}^{AoA})},\notag\\
&\xi_{3}=e^{-j\frac{2\pi}{\lambda}d_{3}}, \xi_{4}=e^{-j\frac{2\pi}{\lambda}d_{4}}, \xi_{5}=e^{-j\frac{2\pi}{\lambda}d_{5}}, \notag
\end{align}
where $e^{j\frac{2\pi}{\lambda}\epsilon_R\cos(\psi_{4}^{AoA})}$ and $e^{j\frac{2\pi}{\lambda}\epsilon_R\cos(\psi_{5}^{AoA})}$ follow from the array response of the TA and the RAs, respectively. As shown in Fig. \ref{Fig simulation model}, the reflective elements of the IS are deployed on the $y$-$z$ plane, and the antennas of the DFBS are deployed along the $x$-axis direction. The steering vectors  are given by
\begin{align}
&\bm{\xi}_I(\vartheta)\! =\!\! \big(\mathbf{1}_{\!N_z}\!\!\otimes\!(1,e^{-j\frac{2\pi}{\lambda}\epsilon_I\cos(\vartheta)},\ldots,e^{-j\frac{2\pi}{\lambda}\epsilon_I(N_y-1)\cos(\vartheta)})\!\big)^{\!\top},\notag\\
&\bm{\xi}_R(\vartheta) = (1,e^{-j\frac{2\pi}{\lambda}\epsilon_R\cos(\vartheta)},\ldots,e^{-j\frac{2\pi}{\lambda}\epsilon_R(M-1)\cos(\vartheta)})^\top,\notag
\end{align}
where $N_y$ and $N_z$ denote the number of reflective elements per row along the $y$-axis direction and  per column along the $z$-axis direction, respectively, $\epsilon_R$ and $\epsilon_I$ denote the antenna spacing and reflective element spacing respectively.

Assume that the IS has $N=300$ reflective elements with $N_y=30$ and $N_z=10$. The BS has $M=4$ RAs. The transmit power and the noise power are $P=10$ dBm and $\sigma^2=-110$ dBm, respectively. The transmission takes $L=1000$ time slots,  and the wavelength is $\lambda = 0.06$ meters. The antennas and reflective elements are half-wavelength spacing, i.e., $\epsilon_R\!=\!\epsilon_I=0.03$ meters.

We compare the proposed scheme with two benchmarks. The first is an exhaustive approach that computes the optimal $\nu$ that maximizes the distance between the actual AoA $\psi_{4}^{AoA}$ and the AoA in \eqref{eqn first change 3}. The second is called the max-inner method that optimizes $\nu$ by maximizing $|(\mathbf{h}^H_{4}+\xi_{5}^*\bm{\xi}_R(\psi_{5}^{AoA})^H\nu^*)\bm{\xi}_R(\psi_{5}^{AoA})|^2$ subject to constraints \eqref{eqn transform problem second 2} and \eqref{eqn transform problem second 3}. Notice that the optimal $\nu$ for the max-inner method can be obtained through an algorithm similar to Algorithm \ref{alg nu}.

\begin{figure}[!t]
\begin{minipage}[t]{0.49\linewidth}
\centering
\includegraphics[height=3.5cm]{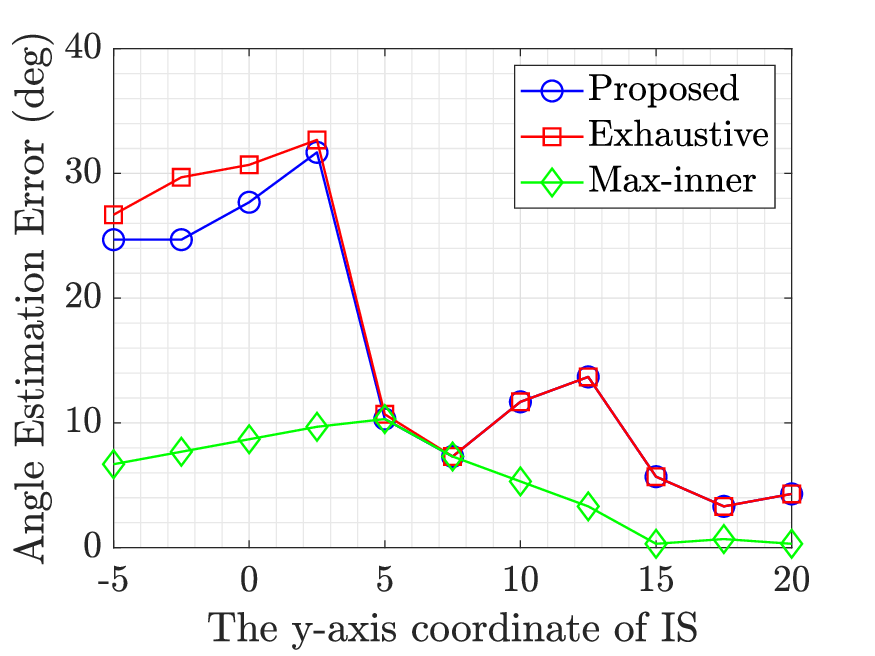}
\subcaption{ }
\end{minipage}
\begin{minipage}[t]{0.49\linewidth}
\centering
\includegraphics[height=3.5cm]{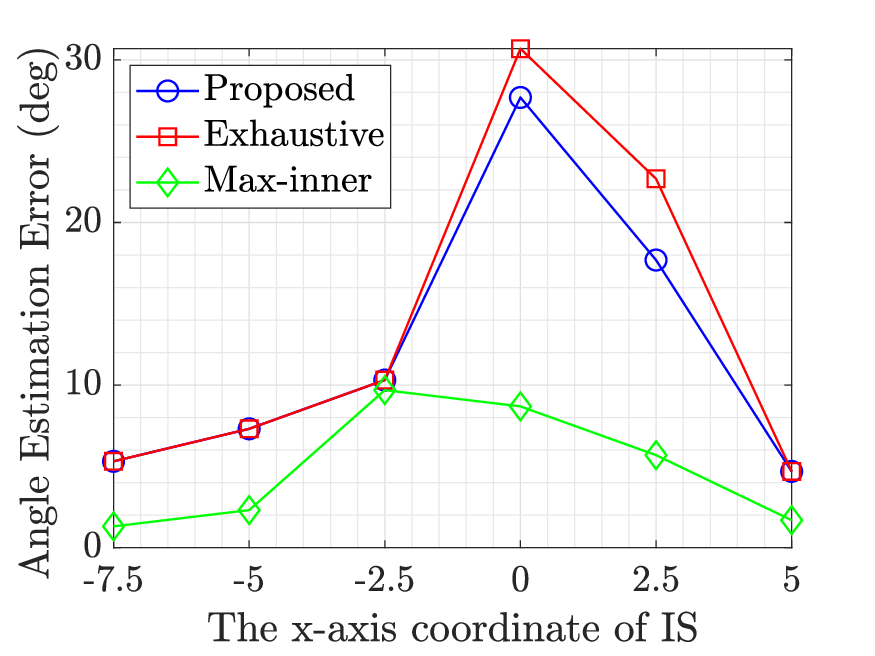}
\subcaption{ }
\end{minipage}
\caption{The AoA estimation error versus the IS location.}
\label{Fig Errorvslocation}
\vspace{-1em}
\end{figure}

Fig. \ref{Fig Errorvslocation} shows the angle estimation error (in degrees) between the actual and estimated AoAs across different methods versus the IS location. The IS is positioned at $(0,y,0)$ and $(x,0,0)$ in Figs. \ref{Fig Errorvslocation}(a) and \ref{Fig Errorvslocation}(b), respectively. Fig. \ref{Fig Errorvslocation} illustrates that the error of the proposed method is close to that of the exhaustive approach, and typically larger than that of the max-inner method. Notice that the error varies with the IS position since the changes in AoA and AoD will affect the beam pattern of the signals \cite{Liu2022}.


\begin{figure}[!t]
\begin{minipage}[t]{0.49\linewidth}
\centering
\includegraphics[height=3.5cm]{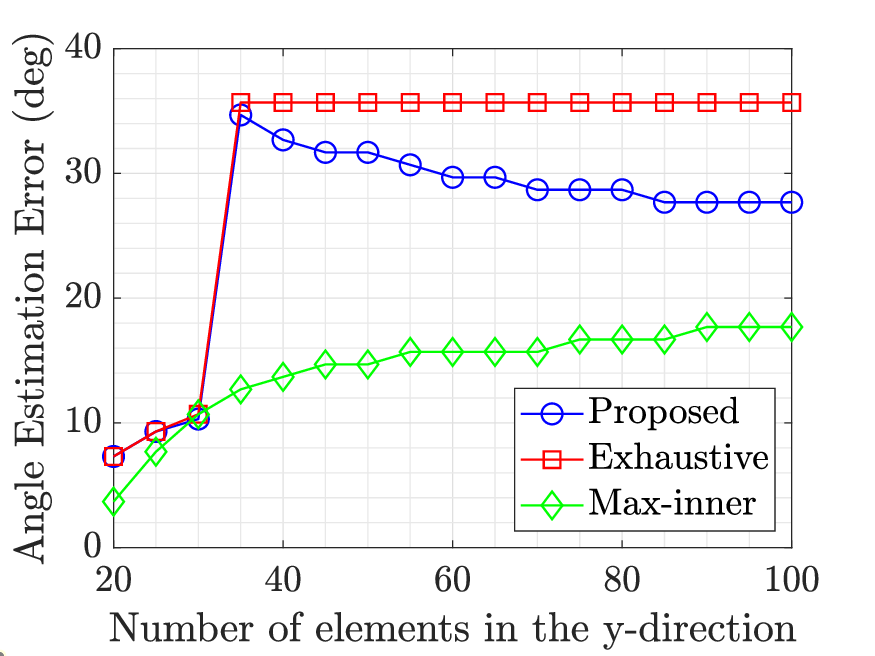}
\subcaption{angle estimation error vs. $N_y$}
\end{minipage}
\begin{minipage}[t]{0.49\linewidth}
\centering
\includegraphics[height=3.5cm]{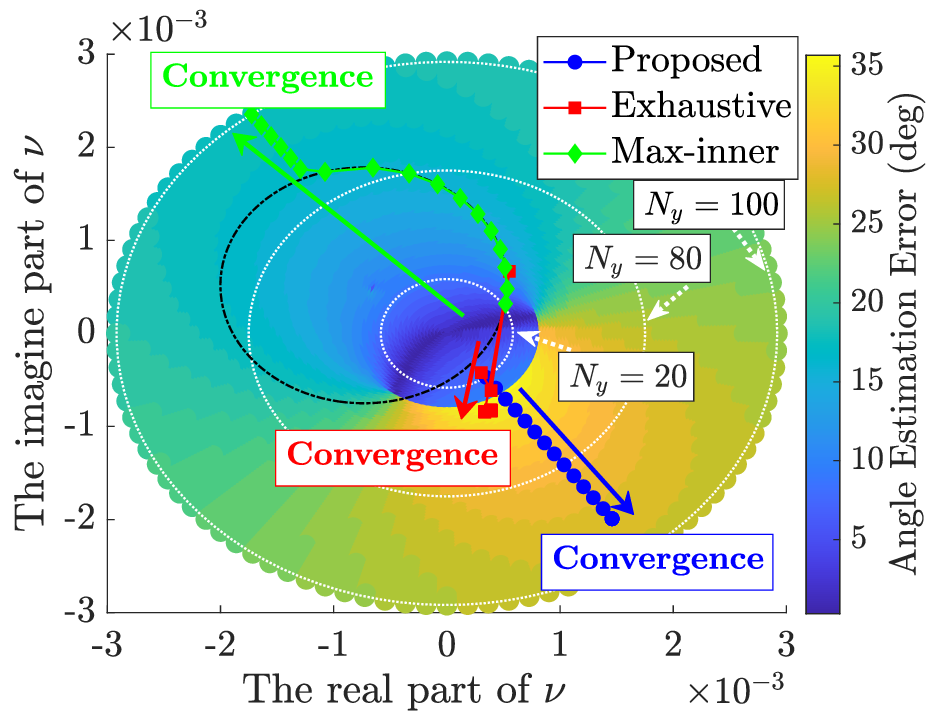}
\subcaption{computing $\nu$}
\end{minipage}
\caption{The angle estimation error vs. the number of reflective elements along the $y$-axis direction.}
\label{Fig ErrorvsNy}
\end{figure}

Fig. \ref{Fig ErrorvsNy} illustrates the angle estimation error (in degrees) and the corresponding $\nu$ as a function of $N_y$. Fig. \ref{Fig ErrorvsNy}(a) shows that the error of the proposed method  is larger than that of the max-inner method, and is close to that of the exhaustive approach for small $N_y$. As $N_y$ increases, the errors from all methods converge, and the error of the proposed method decreases. This trend is explained in Fig. \ref{Fig ErrorvsNy}(b) which plots the  $\nu$'s computed from different methods. The black and white circles represent the boundaries set by constraints \eqref{eqn transform problem second 2} and \eqref{eqn transform problem second 3}, respectively. The feasible region for each $N_y$ is within the corresponding white circle, excluding the region enclosed by the black circle. As $N_y$ increases, the white circle enlarges, leading the optimal $\nu$ for each method towards the unconstrained optimal. For low $N_y$, the optimal $\nu$ of the proposed method (blue points) aligns closely with that of the exhaustive method. However, with higher $N_y$, these blue points diverge towards their convergence point, deviating from the optimal $\nu$'s of the exhaustive method and diminishing the performance. Nonetheless, the proposed method still surpasses the max-inner method.

\begin{figure}[t]
\centering
\includegraphics[height=3.5cm]{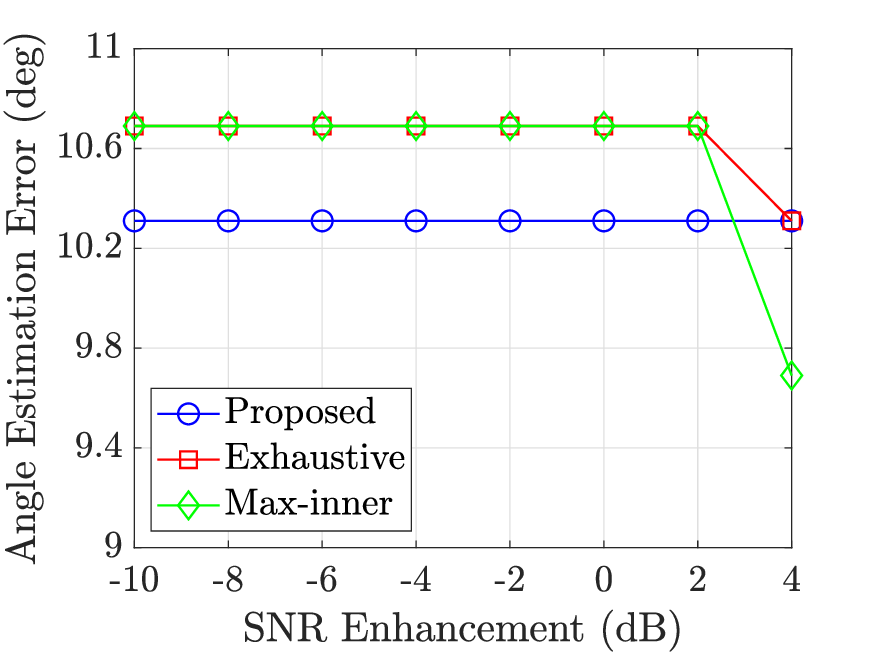}
\caption{The angle estimation error vs. the SNR enhancement.}\label{Fig ErrorvsSNR}
\end{figure}

Fig. \ref{Fig ErrorvsSNR} shows the angle estimation error vs. the SNR enhancement which is based on the without IS case. The error of the proposed method aligns with the global maximum and exceeds that of the max-inner method at high SNR enhancements, but is slightly less when SNR enhancement is below 2 dB. It remains constant with SNR changes because the optimal solution $\nu_6$ is independent of $\eta$.

\section{Conclusion}

This paper considers using the IS to modify the wireless environment to enable radar stealth in the presence of unauthorized ISAC. We propose the notion of inner-product utility to quantify the capability of the DFBS to sense the target AoA. Although the resulting anti-sensing problem with the communication SNR constraint is nonconvex, we show that it admits a closed-form solution from a geometric viewpoint. According to the simulation results, the proposed phase shifting algorithm for the IS can significantly hinder the adversarial sensing while limiting the side effects on communications.

\bibliographystyle{IEEEtran}     
\bibliography{IEEEabrv,strings}

\end{document}